\documentclass[draftclsnofoot,peerreview,english]{IEEEtran}
\usepackage[T1]{fontenc}
\usepackage[latin9]{inputenc}
\usepackage{amsmath}
\usepackage{amssymb}
\usepackage{mathrsfs}
\usepackage{bbm}
\usepackage{mathtools}
\usepackage{empheq}
\usepackage{graphicx}
\usepackage{amsthm}	
\usepackage{hyperref}
\usepackage{cite}
\usepackage{multirow}


\usepackage{babel}

\newtheorem{thm}{Theorem}				
\newtheorem{prop}[thm]{Proposition}		
\newtheorem{lem}[thm]{Lemma}			

\DeclareMathOperator{\pr}{\mathrm{Pr}} 		

\begin{document}

\title{The Benefit of Encoder Cooperation\\
in the Presence of State Information}

\author{Parham~Noorzad,~\IEEEmembership{Student Member,~IEEE,} 
Michelle Effros,~\IEEEmembership{Fellow,~IEEE,}\\ and Michael
Langberg,~\IEEEmembership{Senior~Member,~IEEE}%
\thanks{This material is based upon work supported by the 
National Science Foundation under Grant Numbers 1527524 and 1526771.}%
\thanks{P. Noorzad and M. Effros are with the California
Institute of Technology, Pasadena, CA 91125 USA
(emails: parham@caltech.edu, effros@caltech.edu). }%
\thanks{M. Langberg is with the State
University of New York at Buffalo, Buffalo, NY 14260 USA
(email: mikel@buffalo.edu).}}

\maketitle

\begin{abstract}
In many communication networks, 
the availability of channel state information 
at various nodes provides an opportunity for 
network nodes to work together, or ``cooperate.''
This work studies the benefit of cooperation in the 
multiple access channel with a cooperation facilitator,
distributed state information at the encoders, 
and full state information available at the decoder. 
Under various causality constraints, sufficient conditions are 
obtained such that encoder cooperation through the facilitator results in
a gain in sum-capacity that has infinite slope in the information
rate shared with the encoders. This result extends the prior work of the authors
on cooperation in networks where none of the nodes have access
to state information.
\end{abstract}


\section{Introduction} \label{sec:intro}

In cooperative strategies, various network nodes work together towards
a common goal. Previous work \cite{kUserMAC} shows that under a model of cooperation that 
incorporates a ``cooperation facilitator'' (CF)---a node that receives rate-limited
information from the encoders of a multiple access channel (MAC) and 
sends rate-limited information back---even a very low rate cooperation 
between the MAC encoders can 
vastly increase the total rate that can be delivered through the MAC. In fact,
if we measure cost as the number of bits the CF shares with the encoders and
the benefit as the gain in sum-capacity, then for some MACs, the cost-benefit curve
has infinite slope in the limit of low cost. This paper extends the exploration of 
cooperation beyond the networks of \cite{kUserMAC} to examine the cost-benefit
tradeoff of cooperation in networks where state information is present at some
nodes. 

Networks where state information is available at some nodes appear
in many applications, 
including wireless channels with fading \cite{GoldsmithVaraiya,TseHanly}, 
cognitive radios \cite{SomekhBaruchEtAl}, 
and computer memory with defects \cite{HeegardElGamal}.
Depending on the application at hand, channel state information
may be either fully available at all network nodes or available
in a distributed manner; in the latter case, 
each node has access to a component or
a function of the state sequence. Furthermore, the state 
information may be available non-causally, or alternatively, 
may be subject to causality constraints. 
For example, when state information 
models fading effects in wireless communication \cite{GoldsmithVaraiya}, 
the transmitters' knowledge of
state information is strictly causal or causal. On the other hand,
when the state sequence models a signal that the transmitter sends
to another receiver, then the state
sequence is available non-causally at the transmitter \cite{Costa}. 

In this work, we study the advantage of \emph{encoder cooperation}
in the setting of networks with state information. In this context,
network nodes work together to increase transmission rates---not only
by sharing message information, but also by sharing state information.
(See Figure \ref{fig:network}.)
As an example of message and state cooperation, 
Permuter, Shamai, and Somekh-Baruch \cite{PermuterEtAl} find the capacity
region of the MAC with encoder cooperation 
under the assumption that distributed, non-causal 
state information is available at the encoders and full state 
information is available at the decoder. 
As their cooperation model, the authors use
a special case of the Willems conferencing model \cite{WillemsMAC}, 
originally defined for MACs in the absence of state information. 

Indirect forms of cooperation, in the presence of state information,
are also considered in the literature. 
Cemal and Steinberg \cite{CemalSteinberg} study a model where a 
central state-encoder sends rate limited versions of non-causal 
state information to each encoder, while the decoder has access 
to full state information. 

Here we study cooperation under the CF model. 
In this model, encoders cooperate indirectly as in \cite{kUserMAC}, 
rather than directly, as in \cite{WillemsMAC}. 
The CF enables both message and state cooperation; 
this proves crucial to the cooperation
gain obtained through a CF, which we next describe in more detail. 
 
In earlier work \cite{kUserMAC},
we exhibit single-letter conditions on the channel transition matrix of the 
MAC which guarantee an infinite slope in sum-capacity
as a function of the capacities of the CF output edges; the
additive Gaussian MAC \cite[p. 544]{CoverThomas} 
is an important example of a scenario where the infinite
slope phenomenon occurs.    
In this work, we characterize channels for which the 
cooperation gain has an infinite slope in the presence of state
information (Section \ref{sec:resultMACwState}); interestingly, this 
includes channels for which the infinite slope phenomenon did
not arise in the absence of state information.\footnote{As an 
example, consider the MAC $Y=X_1+X_2+S$ (mod 3), where
$S$ is uniform on $\{0,1,2\}$, $X_1$ and $X_2$ are binary, and
$Y$ is ternary. The infinite slope sum-capacity gain is achievable 
when the decoder has full knowledge of $S$, but no
sum-capacity gain is possible when it does 
not have access to $S$.}

For state information at the encoders we consider four cases:
(i) no state information,
(ii) strictly causal state information,
(iii) causal state information, and
(iv) non-causal state information. In case (i), the CF is  
used for sharing message information (a strategy here called 
``message cooperation'') since no state information is
available at the encoders. In cases (ii)-(iv), the CF 
enables both message and state cooperation. However,
here we study message and state cooperation only in case (iv);
in this case we show that the use of joint message and state
cooperation leads to a weaker sufficient condition 
for an infinite-slope
gain compared to the sole use of message cooperation.
Whether in cases (ii) and (iii), the use of joint message 
and state cooperation likewise leads to a weaker sufficient condition 
for an infinite-slope gain compared to message cooperation alone,
remains an open problem. 

Throughout, we assume that any state information
available at the encoders is distributed; that is, 
we assume $S=(S_1,S_2)$, where for $i\in\{1,2\}$,
$S_i$ is available at encoder $i$.
As we do not make any assumptions regarding the 
dependence between $S_1$ and $S_2$, our results apply to
the limiting cases of independent states (i.e., independent $S_1$ and $S_2$)
and common state (i.e., $S_1=S_2$). 

Since the decoder starts the decoding process only after receiving
all the output symbols in a given transmission block, 
causality constraints at the decoder do not impose limitations
on the availability of state information. 
Thus we may assume that the decoder 
either has full state information or no state information.
Here we focus on the former scenario.
Jafar \cite{Jafar} provides the capacity region of the MAC with distributed 
independent (causal or non-causal) state information at the encoders
and full state information at the decoder.  
The capacity region is unknown when the encoders have access
to state information but the decoder does not 
\cite{LapidothSteinberg1, LapidothSteinberg2}. 

\begin{figure} 
	\begin{center}
		\includegraphics[scale=0.3]{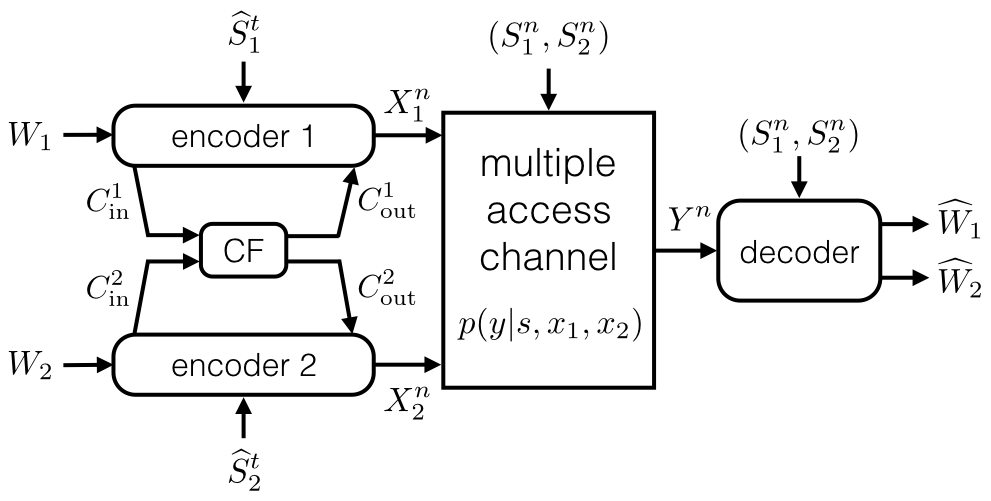}
		\caption{The network studied in this work consists of a pair of encoders
		communicating, with the help of a CF, to a decoder through a
		state-dependent MAC. Full state information is available at the decoder.
		At time $t\in[n]$, partial state information $\widehat{S}_i^t$ 
		is available to encoder $i\in\{1,2\}$.} \label{fig:network}
	\end{center}
\end{figure}

\section{Model} \label{sec:model}
\subsection{Preliminaries} \label{subsec:prelim}
Let $\mathcal{S}_1$, $\mathcal{S}_2$, 
$\mathcal{X}_1$, $\mathcal{X}_2$, and $\mathcal{Y}$
be discrete or continuous alphabets. 
A MAC with input alphabet
$\mathcal{X}_1\times\mathcal{X}_2$, output 
alphabet $\mathcal{Y}$, and state alphabet 
$\mathcal{S}\coloneqq \mathcal{S}_1\times\mathcal{S}_2$ is given by the sequence 
\begin{equation*}
  \Big\{p(s^n)p(y^n|s^n,x_1^n,x_2^n)\Big\}_{n=1}^\infty.
\end{equation*} 
The MAC is said to be memoryless and stationary
if for some $p(s)p(y|s,x_1,x_2)$ and all positive integers $n$, 
\begin{equation*}
  p(s^n)p(y^n|s^n,x_1^n,x_2^n)
  =\prod_{t=1}^n p(s_t)p(y_t|s_t,x_{1t},x_{2t}).
\end{equation*}

\subsection{Message Cooperation}
\label{subsec:msgCoop}

In this subsection, we define the capacity region
of a MAC with a CF that enables message cooperation.
We include four scenarios in our definition based on
the availability of state information at the encoders:
no state, strictly causal, causal, and non-causal. We 
assume full state information is available at the 
decoder. In our definition below, for any real number
$x\geq 1$, $[x]$ denotes the set
$\{1,\dots,\lfloor x\rfloor\}$ .

We start by defining a $(2^{nR_1},2^{nR_2},n)$-code
for the MAC with a $(\mathbf{C}_\mathrm{in},\mathbf{C}_\mathrm{out})$-CF, 
cost functions
$b_i:\mathcal{X}_i\rightarrow\mathbb{R}_{\geq 0}$
for $i\in \{1,2\}$, and cost constraints $B_1,B_2\geq 0$. 
The pairs 
$\mathbf{C}_\mathrm{in}=(C_\mathrm{in}^1,C_\mathrm{in}^2)$
and $\mathbf{C}_\mathrm{out}=(C_\mathrm{out}^1,C_\mathrm{out}^2)$
denote the CF input and output edge capacities, respectively.
Encoder $i$, for $i\in\{1,2\}$, is represented by 
$(\varphi_\mathrm{in}^i,(f_{it})_{t=1}^n)$, the CF
is represented by 
$(\varphi_\mathrm{out}^1,\varphi_\mathrm{out}^2)$, and the 
decoder is represented by $g$. These mappings are defined
in the order of their use below. 
For $i\in\{1,2\}$, the transmission from encoder $i$ to the
CF is represented by the mapping 
\begin{equation} \label{eq:encCF}
  \varphi_\mathrm{in}^i: [2^{nR_i}]\rightarrow [2^{nC_\mathrm{in}^i}]
\end{equation}
and the transmission from the CF to encoder $i$ is represented by
\begin{equation*}
  \varphi_\mathrm{out}^i:[2^{nC_\mathrm{in}^1}]\times 
  [2^{nC_\mathrm{in}^2}]\rightarrow
  [2^{nC_\mathrm{out}^i}].
\end{equation*}
For simplicity, the transmissions to and from the CF
occur prior to the transmission of codewords over 
the channel. 

At time $t\in [n]$, for $i\in\{1,2\}$,
the transmission of encoder $i$ over 
the channel is represented by the mapping 
\begin{equation} \label{eq:encoder}
  f_{it} : [2^{nR_i}]\times[2^{nC_\mathrm{out}^i}]
  \times\widehat{\mathcal{S}}_i^t\rightarrow\mathcal{X}_i.
\end{equation}
Here $\widehat{S}_i^t$ represents any knowledge about the state gathered
by encoder $i$ in times $\{1,\dots,t\}$. Let $*$ be a symbol not in 
$\mathcal{S}_1\cup\mathcal{S}_2$. For $i\in\{1,2\}$ and $t\in [n]$, 
we have
\begin{equation*}
  \widehat{S}_{it}\coloneqq 
  \begin{cases}
  * &\text{no state information}\\
  S_{i(t-1)} &\text{strictly causal}\\
  S_{it} &\text{causal}\\
  S_i^n &\text{non-causal}.
  \end{cases}
\end{equation*}
For every message pair $(w_1,w_2)$, the codeword of
encoder $i$ is required to satisfy the cost constraint
\begin{equation} \label{eq:cost}
  \sum_{t=1}^n \mathbb{E}b_i\Big[f_{it}\big(w_i,
  \varphi_\mathrm{out}^i(\varphi_\mathrm{in}^1(w_1),
  \varphi_\mathrm{in}^2(w_2)),\widehat{S}_i^t\big)\Big]\leq B_i. 
\end{equation}
The decoder has full state information and is represented by the mapping
\begin{equation*}
  g:\mathcal{S}^n\times\mathcal{Y}^n
	\rightarrow [2^{nR_1}]\times [2^{nR_2}].
\end{equation*}
The average probability of error is given by
\begin{equation*}
  P_e^{(n)}=\pr\Big\{g(S^n,Y^n)\neq (W_1,W_2)\Big\},
\end{equation*}
where $(W_1,W_2)$ is uniformly distributed over 
$[2^{nR_1}]\times [2^{nR_2}]$.
A rate pair $(R_1,R_2)$ is achievable if there exists a sequence of 
$(2^{nR_1},2^{nR_2},n)$-codes with $P_e^{(n)}\rightarrow 0$ as 
$n\rightarrow\infty$. We use subscript $\tau\in\{0,T-1,T,\infty\}$
to specify the dependence of the capacity region and 
sum-capacity on the availability of state information at the encoders.
The following table makes this dependence clear.

\begin{center}
\begin{tabular}{ c | c }			
  $\tau$ & encoder state information \\
  \hline
  0 & none \\
  $T-1$ & strictly causal \\
  $T$ & causal \\
  $\infty$ & non-causal \\ 
\end{tabular}
\end{center}

The capacity region 
$\mathscr{C}_\tau(\mathbf{C}_\mathrm{in},\mathbf{C}_\mathrm{out})$
is given by the closure of all achievable rate pairs. The sum-capacity,
denoted by $C_\tau(\mathbf{C}_\mathrm{in},\mathbf{C}_\mathrm{out})$,
is defined as
\begin{equation} \label{eq:sumcapacity}
  C_\tau(\mathbf{C}_\mathrm{in},\mathbf{C}_\mathrm{out})
  \coloneqq \max_{\mathscr{C}_\tau(\mathbf{C}_\mathrm{in},\mathbf{C}_\mathrm{out})}
  (R_1+R_2).
\end{equation}
For example, $\mathscr{C}_T(\mathbf{C}_\mathrm{in},\mathbf{C}_\mathrm{out})$ and
$C_T(\mathbf{C}_\mathrm{in},\mathbf{C}_\mathrm{out})$ denote the capacity region
and sum-capacity, respectively, of a MAC with a 
$(\mathbf{C}_\mathrm{in},\mathbf{C}_\mathrm{out})$-CF
and distributed causal state information available at the encoders.

\subsection{Message and State Cooperation}
\label{subsec:msgStateCoop}

In the scenario where non-causal state information is available at the encoders, we also
study the benefit of joint message and state cooperation. In the definition of a code 
for the case where non-causal state information is available at the encoders (Subsection 
\ref{subsec:msgCoop}), for $i\in\{1,2\}$, replace (\ref{eq:encCF}) and
(\ref{eq:cost}) with
\begin{equation*} 
  \varphi_\mathrm{in}^i: [2^{nR_i}]\times \mathcal{S}_i^n
  \rightarrow [2^{nC_\mathrm{in}^i}],\text{ and}
\end{equation*}
\begin{equation*} 
  \sum_{t=1}^n \mathbb{E}b_i\Big[f_{it}\big(w_i,
  \varphi_\mathrm{out}^i(\varphi_\mathrm{in}^1(w_1,S_1^n),\varphi_\mathrm{in}^2(w_2,S_2^n)),
  S_i^n\big)\Big]\leq B_i.
\end{equation*}
We denote the capacity region and sum-capacity with 
$\mathscr{C}_{\infty,s}(\mathbf{C}_\mathrm{in},\mathbf{C}_\mathrm{out})$
and $C_{\infty,s}(\mathbf{C}_\mathrm{in},\mathbf{C}_\mathrm{out})$, respectively.
The subscript ``$s$'' indicates the dependence of the cooperation strategy
on the channel state information.

\section{Coding Strategy} \label{sec:codeState}

Here we describe our coding strategies, which are based on random coding arguments. 
Since our aim is to determine conditions sufficient for an infinite 
slope cooperation gain, we specifically focus on coding
strategies that lead to large gains for small cooperation rates such
as the coordination strategy.\footnote{The coordination strategy
\cite{kUserMAC} is the adaptation of Marton's coding strategy
for the broadcast channel \cite{Marton} to the MAC with encoder
cooperation.} 
In particular, in the coding strategies below, 
the CF does not use its rate for 
forwarding message or state information \cite{kUserMAC}, 
since in the cases studied in the literature \cite{WillemsMAC, PermuterEtAl}, 
the gain of such a strategy is at most linear in the cooperation rate. 
We start with message cooperation and 
conclude with message and state cooperation. 

\subsection{Inner Bound for Message Cooperation}

For simplicity, we assume the CF has
access to both messages by setting 
$\mathbf{C}_\mathrm{in}=\mathbf{C}_\mathrm{in}^* 
=(C_\mathrm{in}^{*1},C_\mathrm{in}^{*2})$, where 
$C_\mathrm{in}^{*1}$ and $C_\mathrm{in}^{*2}$ are
sufficiently large. Despite this assumption, 
our main result regarding sum-capacity gain,
Theorem \ref{thm:mainMACwState}, holds for 
any $\mathbf{C}_\mathrm{in}\in\mathbb{R}^2_{>0}$. This is due to
the fact that using time-sharing, as stated in the lemma below,
we can use the inner bounds for $\mathbf{C}_\mathrm{in}^*$
to obtain inner bounds for any
$\mathbf{C}_\mathrm{in}\in\mathbb{R}^2_{>0}$. The proof
appears in Subsection \ref{subsec:CinMuCinStar}. 
\begin{lem} \label{lem:CinMuCinStar}
Consider a memoryless stationary MAC. For any
$(\mathbf{C}_\mathrm{in},\mathbf{C}_\mathrm{out})
\in\mathbb{R}^2_{>0}\times\mathbb{R}^2_{\geq 0}$, 
there exists $\mu>0$, depending only on $\mathbf{C}_\mathrm{in}$,
such that for all $\tau\in\{0,T-1,T,\infty\}$,
\begin{equation*}
  C_\tau(\mathbf{C}_\mathrm{in},
  \mathbf{C}_\mathrm{out})
  -C_\tau(\mathbf{C}_\mathrm{in},\mathbf{0})
  \geq
  \mu\Big(C_\tau(\mathbf{C}_\mathrm{in}^*,
  \mathbf{C}_\mathrm{out})
  -C_\tau(\mathbf{C}_\mathrm{in}^*,\mathbf{0})\Big).
\end{equation*}  
\end{lem} 

We first describe our inner bound for the case where the 
encoders do not have access to state information. In
this case, even though the decoder has access to full
state information, we can obtain a suitable inner bound
by applying results where state information is absent
at both the encoders and the decoder to a modified channel.
Specifically, applying
\cite[Theorem 1]{kUserMAC} to the channel
\begin{equation*}
  \Big(\mathcal{X}_1\times\mathcal{X}_2,
  p(y,s|x_1,x_2),\mathcal{Y}\times\mathcal{S}\Big),
\end{equation*}
where
\begin{equation*}
  p(y,s|x_1,x_2)=p(s)p(y|s,x_1,x_2),
\end{equation*}
gives an inner bound for the channel
$p(y|s,x_1,x_2)$ when full state information is 
available at the decoder. We note that 
applying Lemma \ref{lem:IBnoState} 
together with the outer bound
presented in Subsection \ref{subsec:outerBoundsState}
gives the capacity region in the absence of cooperation
$(\mathbf{C}_\mathrm{out}=\mathbf{0})$ both in the case
where no state information is available at the encoders
and in the case where the state information available
at the encoders is strictly causal. 
\begin{lem} \label{lem:IBnoState}
The set of all rate pairs $(R_1,R_2)$ satisfying 
\begin{align*}
  R_1 &\leq I(X_1;Y|S_1,S_2,X_2)\\
  R_2 &\leq I(X_2;Y|S_1,S_2,X_1)\\
  R_1+R_2 &\leq I(X_1,X_2;Y|S_1,S_2)
\end{align*}
for some distribution $p(x_1)p(x_2)$ with  
\begin{equation*}
  I(X_1;X_2)\leq C_\mathrm{out}^1+C_\mathrm{out}^2
\end{equation*}
and $\mathbb{E}[b_i(X_i)]\leq B_i$ for $i\in\{1,2\}$, 
is contained in 
$\mathscr{C}_0(\mathbf{C}_\mathrm{in}^*,\mathbf{C}_\mathrm{out})$. 
\end{lem}

In the case where the encoders have access to causal state information,
the codeword transmitted by an encoder can depend both on its 
message and the present state information. 
Lemma \ref{lem:IBcausal} provides an inner bound for the capacity region
in this scenario. In this inner bound, for $i\in\{1,2\}$, $U_i$ encodes
the message of encoder $i$ in addition to the information it receives
from the CF. Note that this inner bound is tight when 
$\mathbf{C}_\mathrm{out}=\mathbf{0}$, even if \emph{non-causal}
state information is available at the encoders. (See Subsection
\ref{subsec:IBcausal} for the proof of this lemma and Subsection
\ref{subsec:outerBoundsState} for the corresponding outer bound
in the absence of cooperation.)
\begin{lem} \label{lem:IBcausal}
The set of all rate pairs satisfying 
\begin{align*}
  R_1 &\leq I(U_1;Y|S_1,S_2,U_2)\\
  R_2 &\leq I(U_2;Y|S_1,S_2,U_1)\\
  R_1+R_2 &\leq I(U_1,U_2;Y|S_1,S_2)
\end{align*}
for some distribution $p(u_1,u_2)p(x_1|u_1,s_1)p(x_2|u_2,s_2)$
with\footnote{As we show in Subsection \ref{subsec:IBcausal},
we can choose $X_1$ and $X_2$ to be deterministic functions
of $(U_1,S_1)$ and $(U_2,S_2)$, respectively.}  
\begin{equation*}
  I(U_1;U_2)\leq C_\mathrm{out}^1+C_\mathrm{out}^2
\end{equation*}
and $\mathbb{E}[b_i(X_i)]\leq B_i$ for $i\in\{1,2\}$, 
is contained in 
$\mathscr{C}_T(\mathbf{C}_\mathrm{in}^*,\mathbf{C}_\mathrm{out})$. 
\end{lem}

\subsection{Inner Bound for Message and State Cooperation}

As discussed in Subsection \ref{subsec:msgStateCoop}, 
we only consider
message and state cooperation in the scenario 
where non-causal
state information is available at the encoders. 

Here we assume that the state alphabet 
$\mathcal{S}=\mathcal{S}_1\times\mathcal{S}_2$
is discrete and $H(S_1,S_2)$ is finite. Furthermore,
we assume the CF not only has access to both messages, 
but also knows the state sequences $S_1^n$ and $S_2^n$;
equivalently, we set
$\mathbf{C}_\mathrm{in}=\mathbf{\bar{C}}_\mathrm{in}
=(\bar{C}_\mathrm{in}^1,\bar{C}_\mathrm{in}^2)$, where 
$\bar{C}_\mathrm{in}^1$ and $\bar{C}_\mathrm{in}^2$ are  
sufficiently large. A lemma analogous to Lemma 
\ref{lem:CinMuCinStar} holds in this case.

\begin{lem} 
Fix a memoryless stationary MAC. For any
$(\mathbf{C}_\mathrm{in},\mathbf{C}_\mathrm{out})
\in\mathbb{R}^2_{>0}\times\mathbb{R}^2_{\geq 0}$, 
there exists $\mu>0$, depending only on $\mathbf{C}_\mathrm{in}$,
such that 
\begin{equation*}
 C_{(\infty,s)}(\mathbf{C}_\mathrm{in},
 \mathbf{C}_\mathrm{out})
 -C_{(\infty,s)}(\mathbf{C}_\mathrm{in},\mathbf{0})
 \geq
 \mu\Big(C_{(\infty,s)}(\mathbf{\bar{C}}_\mathrm{in},
 \mathbf{C}_\mathrm{out})
 -C_{(\infty,s)}(\mathbf{\bar{C}}_\mathrm{in},\mathbf{0})\Big),
\end{equation*}  
\end{lem} 

We next describe our coding strategy for the MAC with
message and state cooperation. 

\textbf{Codebook Generation.} 
Choose a distribution $p(x_1,x_2|s_1,s_2)$. 
For $i\in\{1,2\}$, $w_i\in [2^{nR_i}]$, 
$z_i\in [2^{nC_\mathrm{out}^i}]$, $s_i^n\in\mathcal{S}_i^n$, generate
$X_i^n(w_i,z_i|s_i^n)$ i.i.d.\ according to the distribution
\begin{equation*}
  \pr\Big\{X_i^n(w_i,z_i|s_i^n)=x_i^n\Big| S_i^n=s_i^n\Big\}
  =\prod_{t=1}^n p(x_{it}|s_{it}).
\end{equation*}

\textbf{Encoding.} The CF, having access to $(w_1,w_2)$ 
and $(S_1^n,S_2^n)$, looks for a pair 
$(Z_1,Z_2)\in [2^{nC_\mathrm{out}^1}]\times[2^{nC_\mathrm{out}^2}]$ satisfying
\begin{equation} \label{eq:typical}
  \Big(S_1^n,S_2^n,X_1^n(w_1,Z_1|S_1^n),
  X_2^n(w_2,Z_2|S_2^n)\Big)\in A_\delta^{(n)},
\end{equation}
where $A_\delta^{(n)}$ is the weakly typical set with respect to the distribution
$p(s_1,s_2)p(x_1,x_2|s_1,s_2)$. 
If there is more than one such pair, the CF chooses the smallest pair
according to the lexicographical order. If there is no such pair, it sets
$(Z_1,Z_2)=(1,1)$. 
The CF sends $Z_i$ to encoder $i$ for $i\in\{1,2\}$.  
Encoder $i$ transmits $X_i^n(w_i,Z_i|S_i^n)$ over $n$ uses of the channel.

Using \cite[p. 130, Lemma A.1.1]{Noorzad}, it follows
that as $n$ goes to infinity, the probability that a pair 
$(Z_1,Z_2)$ satisfying (\ref{eq:typical}) exists goes
to one if
\begin{align*}
  C_\mathrm{out}^1 &>
  H(X_1|S_1)-H(X_1|S_1,S_2)+24\delta\\
  C_\mathrm{out}^2 &>
  H(X_2|S_2)-H(X_2|S_1,S_2)+24\delta\\
  C_\mathrm{out}^1+C_\mathrm{out}^2 &>
  H(X_1|S_1)+H(X_2|S_2)-H(X_1,X_2|S_1,S_2)
  +6\delta.
\end{align*}

\textbf{Decoding.} Once the decoder receives $Y^n$, using $(S_1^n,S_2^n)$,
it looks for a pair $(\widehat{w}_1,\widehat{w}_2)$ that satisfies
\begin{equation*}
  \Big(S_1^n,S_2^n,X_1^n(\widehat{w}_1,\widehat{Z}_1|S_1^n),
  X_2^n(\widehat{w}_2,\widehat{Z}_2|S_2^n),Y^n\Big)
  \in A_\epsilon^{(n)}.
\end{equation*}
Here $A_\epsilon^{(n)}$ is the weakly typical set with respect to the distribution
$p(s_1,s_2)p(x_1,x_2|s_1,s_2)p(y|s_1,s_2,x_1,x_2)$.
If there is no such pair, or there is such a pair but it is not unique, the
decoder sets $(\widehat{w}_1,\widehat{w}_2)=(1,1)$.

The error analysis of the above coding scheme leads 
to the following lemma, which provides an inner bound for 
$\mathscr{C}_{\infty,s}(\mathbf{\bar{C}}_\mathrm{in},\mathbf{C}_\mathrm{out})$.
\begin{lem} \label{lem:IBnoncausal}
The set of all rate pairs satisfying 
\begin{align*}
  R_1 &\leq I(X_1;Y|S_1,S_2,X_2)\\
  R_2 &\leq I(X_2;Y|S_1,S_2,X_1)\\
  R_1+R_2 &\leq I(X_1,X_2;Y|S_1,S_2)
\end{align*}
for some distribution $p(x_1,x_2|s_1,s_2)$ with  
\begin{align*}
  C_\mathrm{out}^1 &\geq H(X_1|S_1)-H(X_1|S_1,S_2)\\
  C_\mathrm{out}^2 &\geq H(X_2|S_2)-H(X_2|S_1,S_2)\\ 
  C_\mathrm{out}^1+C_\mathrm{out}^2
  &\geq H(X_1|S_1)+H(X_2|S_2)-H(X_1,X_2|S_1,S_2)
\end{align*}
and $\mathbb{E}[b_i(X_i)]\leq B_i$ for $i\in\{1,2\}$, 
is contained in 
$\mathscr{C}_{\infty,s}(\mathbf{\bar{C}}_\mathrm{in},\mathbf{C}_\mathrm{out})$. 
\end{lem}

\section{Main Result} \label{sec:resultMACwState}

Our main result describes conditions on a MAC
that, if satisfied,
for every fixed $\mathbf{C}_\mathrm{in}\in\mathbb{R}^2_{>0}$,
guarantee an infinite slope in sum-capacity as
a function of $\mathbf{C}_\mathrm{out}$. As sum-capacity depends
on the availability of state information at the encoders, so
do our conditions. The proof appears in 
Subsection \ref{subsec:mainNoncausal}. 

\begin{thm} \label{thm:mainMACwState}
Let $\mathcal{S}$, $\mathcal{X}_1$, $\mathcal{X}_2$,
and $\mathcal{Y}$ be finite sets. 
For any $\tau\in\{0,T-1,T,\infty,(\infty,s)\}$, any MAC in 
$\mathcal{C}_\tau(\mathcal{S},\mathcal{X}_1,\mathcal{X}_2,\mathcal{Y})$,
and any $(\mathbf{C}_\mathrm{in},\mathbf{v})
\in\mathbb{R}^2_{>0}\times\mathbb{R}^2_{>0}$,
\begin{equation*}
  \lim_{h\rightarrow 0^+}
  \frac{C_\tau(\mathbf{C}_\mathrm{in},h\mathbf{v})
  -C_\tau(\mathbf{C}_\mathrm{in},\mathbf{0})}{h}=\infty.
\end{equation*}
\end{thm}

We next specifically define 
$\mathcal{C}_\tau(\mathcal{S},\mathcal{X}_1,\mathcal{X}_2,\mathcal{Y})$
for each subscript $\tau\in\{0,T-1,T,\infty,(\infty,s)\}$;
as defined in Section \ref{sec:model}, $\tau$ specifies 
the availability of state information at the encoders. Note that the 
definition of $C_\tau$ provides \emph{sufficient} conditions for a 
large cooperation gain; these conditions may not be necessary. 

In our descriptions below, all mentioned distributions 
satisfy 
\begin{equation*}
  \mathbb{E}\big[b_i(X_i)\big]
  \leq B_i\text{ for }i\in\{1,2\}.
\end{equation*}

\subsection{Message Cooperation} 

In this subsection, we define classes of MACs which exhibit
a large message cooperation gain as described in Theorem
\ref{thm:mainMACwState}.

\textbf{No state information.} A MAC is in 
$\mathcal{C}_0(\mathcal{S},\mathcal{X}_1,\mathcal{X}_2,\mathcal{Y})$
if 

(i) for some $p_0(x_1)p_0(x_2)$ that satisfies
\begin{equation*}
  I_0(X_1,X_2;Y|S)
  =\max_{p(x_1)p(x_2)}I(X_1,X_2;Y|S),
\end{equation*}
there exists $p_1(x_1,x_2)$ that satisfies
\begin{equation*}
  I_1(X_1,X_2;Y|S)+
  \mathbb{E}\Big[D\big(p_1(y|S)\|p_0(y|S)\big)\Big]
  > I_0(X_1,X_2;Y|S),\text{ and}
\end{equation*}

(ii) $\mathrm{supp}(p_1(x_1,x_2))
\subseteq\mathrm{supp}(p_0(x_1)p_0(x_2))$, where 
``supp'' denotes the support.

Intuitively, condition (i) ensures that our channel has the 
property that the dependence created through message cooperation increases
sum-capacity. Condition (ii) allows the CF to use a small rate 
(i.e., small $\mathbf{C}_\mathrm{out}$) to help the encoders, whose
codewords are generated according to $p_0(x_1)p_0(x_2)$, to transmit
codewords whose distribution is sufficiently close to $p_1(x_1,x_2)$
to achieve a large gain in sum-capacity. 

\textbf{Strictly causal state information.}
As mentioned in Section \ref{sec:codeState},
the availability of strictly causal
state information at the encoders of a MAC without cooperation does not enlarge
the capacity region, thus we set  
$\mathcal{C}_{T-1}(\mathcal{S},\mathcal{X}_1,\mathcal{X}_2,\mathcal{Y})
\coloneqq \mathcal{C}_0(\mathcal{S},\mathcal{X}_1,\mathcal{X}_2,\mathcal{Y})$.

\textbf{Causal state information.}
A MAC is in 
$\mathcal{C}_T(\mathcal{S},\mathcal{X}_1,\mathcal{X}_2,\mathcal{Y})$
if 

(i) for some 
$p_0(x_1|s_1)p_0(x_2|s_2)$ that satisfies
\begin{equation*}
  I_0(X_1,X_2;Y|S)
  =\max_{p(x_1|s_1)p(x_2|s_2)}I(X_1,X_2;Y|S),
\end{equation*}
there exist alphabets $\mathcal{U}_1$, $\mathcal{U}_2$, 
distributions $p_0(u_1)p_0(u_2)$ and $p_1(u_1,u_2)$, and
mappings $f_i:\mathcal{U}_i\times\mathcal{S}_i
\rightarrow\mathcal{X}_i$ for $i\in\{1,2\}$ such that
\begin{align}
  &p_0(x_1|s_1)p_0(x_2|s_2)
  =\sum_{u_1,u_2}
  p_0(u_1)p_0(u_2)\mathbf{1}\big\{x_1=f_1(u_1,s_1)\big\}
  \mathbf{1}\big\{x_2=f_2(u_2,s_2)\big\},\label{eq:p0XSinTermsU}\\
  &I_1(U_1,U_2;Y|S)+
  \mathbb{E}\Big[D\big(p_1(y|S)\|p_0(y|S)\big)\Big]
  > I_0(U_1,U_2;Y|S),\text{ and}\label{eq:CTI1I0}
\end{align}

(ii)
$\mathrm{supp}(p_1(u_1,u_2))\subseteq\mathrm{supp}(p_0(u_1)p_0(u_2))$.

In (\ref{eq:CTI1I0}), the expressions are calculated using
the input distributions 
\begin{align*}
 &p_0(u_1)p_0(u_2)\mathbf{1}\{x_1=f_1(u_1,s_1)\}
  \mathbf{1}\{x_2=f_2(u_2,s_2)\},
 \text{ and}\\
 &p_1(u_1,u_2)\mathbf{1}\{x_1=f_1(u_1,s_1)\}
  \mathbf{1}\{x_2=f_2(u_2,s_2)\}.
\end{align*}

\textbf{Non-causal state information.}
In the absence of cooperation, the capacity region is the same
for regardless of whether the state information at the encoders
is causal or non-causal. Thus, similar to the strictly
causal case, we set
$\mathcal{C}_\infty(\mathcal{S},\mathcal{X}_1,\mathcal{X}_2,\mathcal{Y})
\coloneqq \mathcal{C}_T(\mathcal{S},\mathcal{X}_1,\mathcal{X}_2,\mathcal{Y})$.

\subsection{Message and State Cooperation}

Here we provide sufficient conditions for a large
gain resulting from message and state cooperation. 

\textbf{Non-causal state information.}
A MAC is in 
$\mathcal{C}_{\infty,s}(\mathcal{S},\mathcal{X}_1,\mathcal{X}_2,\mathcal{Y})$
if  

(i) for some $p_0(x_1|s_1)p_0(x_2|s_2)$ that satisfies
\begin{equation*}
  I_0(X_1,X_2;Y|S)
  =\max_{p(x_1|s_1)p(x_2|s_2)}I(X_1,X_2;Y|S),
\end{equation*}
there exists $p_1(x_1,x_2|s_1,s_2)$ that satisfies
\begin{equation*}
  I_1(X_1,X_2;Y|S)+
  \mathbb{E}\Big[D\big(p_1(y|S)\|p_0(y|S)\big)\Big] 
  > I_0(X_1,X_2;Y|S),\text{ and}
\end{equation*}

(ii) for all $(s_1,s_2)\in\mathcal{S}$,
$\mathrm{supp}(p_1(x_1,x_2|s_1,s_2))
\subseteq\mathrm{supp}(p_0(x_1|s_1)p_0(x_2|s_2))$.

\section{Example: Gaussian MAC with Binary Fading} \label{sec:example}

While Theorem \ref{thm:mainMACwState} is stated only for finite
alphabet MACs, the result is not limited to such MACs. Specifically,
for a given MAC, we can use our inner bounds 
described in Section \ref{sec:codeState} to calculate an inner 
bound for sum-capacity and verify the result of 
Theorem \ref{thm:mainMACwState} directly. We next describe
an example of such a MAC.

Consider a MAC that models the wireless 
communication between two encoders
and a decoder in the presence of binary fading.
The input-output relationship of this MAC is given by 
\begin{equation*}
  Y=S_1X_1+S_2X_2+Z,
\end{equation*}
where $(S_1,S_2)$ is uniformly distributed on
$\{0,1\}^2$, and $Z$ is a Gaussian random variable 
with mean zero and variance $N$. In addition, for $i\in\{1,2\}$
we set the cost function $b_i(x)=x^2$ and
cost constraint $B_i=P_i$, so that the cost
constraints correspond to the usual power constraints of the 
Gaussian MAC. 

\begin{prop} \label{prop:gaussianFading}
Consider the Gaussian MAC with binary fading. Fix 
$(\mathbf{C}_\mathrm{in},\mathbf{v})
\in\mathbb{R}^2_{>0}\times\mathbb{R}^2_{>0}$. Then
for all $\tau\in\{0,T-1,T,\infty,(\infty,s)\}$,
\begin{equation*}
  \lim_{h\rightarrow 0^+}
  \frac{C_\tau(\mathbf{C}_\mathrm{in},h\mathbf{v})
  -C_\tau(\mathbf{C}_\mathrm{in},\mathbf{0})}{h}=\infty.
\end{equation*}
\end{prop}

The proof appears in Subsection \ref{subsec:gaussianFading}.

\section{Proofs} \label{sec:MACwStateProofs}

\subsection{Proof of Lemma \ref{lem:CinMuCinStar}}
\label{subsec:CinMuCinStar}

Since $\mathbf{C}_\mathrm{in}\in\mathbb{R}^2_{>0}$,
there exists $\mu\in (0,1)$ such that for $i\in\{1,2\}$,
\begin{equation*}
  \mathbf{C}_\mathrm{in}^i
  \geq \mu\mathbf{C}_\mathrm{in}^{*i}.
\end{equation*}
Then for each $\tau\in\{0,T-1,T,\infty\}$, 
a time-sharing argument shows that 
\begin{align*}
  \mathscr{C}_\tau(\mathbf{C}_\mathrm{in},
  \mathbf{C}_\mathrm{out})
  &\supseteq \mu \mathscr{C}_\tau(\mathbf{C}_\mathrm{in}/\mu,
  \mathbf{C}_\mathrm{out})
  +(1-\mu) \mathscr{C}_\tau(\mathbf{0},\mathbf{C}_\mathrm{out})\\
  &\supseteq \mu \mathscr{C}_\tau(\mathbf{C}_\mathrm{in}^*,
  \mathbf{C}_\mathrm{out})
  +(1-\mu) \mathscr{C}_\tau(\mathbf{0},\mathbf{C}_\mathrm{out}).
\end{align*}
Thus 
\begin{equation*}
  C_\tau(\mathbf{C}_\mathrm{in},
  \mathbf{C}_\mathrm{out})\geq
  \mu C_\tau(\mathbf{C}_\mathrm{in}^*,
  \mathbf{C}_\mathrm{out})
  +(1-\mu) C_\tau(\mathbf{0},\mathbf{C}_\mathrm{out}),
\end{equation*}
which implies
\begin{equation*}
  C_\tau(\mathbf{C}_\mathrm{in},
  \mathbf{C}_\mathrm{out})
  -C_\tau(\mathbf{C}_\mathrm{in},\mathbf{0})\geq
  \mu\Big(C_\tau(\mathbf{C}_\mathrm{in}^*,
  \mathbf{C}_\mathrm{out})
  -C_\tau(\mathbf{C}_\mathrm{in}^*,\mathbf{0})\Big)
\end{equation*}
since
\begin{equation*}
   C_\tau(\mathbf{0},\mathbf{C}_\mathrm{out})
   =C_\tau(\mathbf{0},\mathbf{0})
   =C_\tau(\mathbf{C}_\mathrm{in},\mathbf{0})
   =C_\tau(\mathbf{C}_\mathrm{in}^*,\mathbf{0}).
\end{equation*}

\subsection{Proof of Lemma \ref{lem:IBcausal}} 
\label{subsec:IBcausal}

Fix alphabets $\mathcal{U}_1$ and $\mathcal{U}_2$, and mappings 
\begin{equation*}
  f_i:\mathcal{U}_i\times\mathcal{S}_i
  \rightarrow\mathcal{X}_i\text{ for }
  i\in\{1,2\}.
\end{equation*}
Applying Lemma \ref{lem:IBnoState}, where state information 
is only available at the decoder, to the channel
\begin{equation} \label{eq:channelTransform}
  p(y|s,u_1,u_2)=\sum_{x_1,x_2}
  p(y|s,x_1,x_2)\mathbf{1}
  \big\{x_1=f_1(u_1,s_1)\big\}
  \mathbf{1}\big\{x_2=f_2(u_2,s_2)\big\}
\end{equation}
shows that the set of all rate pairs satisfying
\begin{align*}
  R_1 &\leq I(U_1;Y|S,U_2)\\
  R_2 &\leq I(U_2;Y|S,U_1)\\
  R_1+R_2 &\leq I(U_1,U_2;Y|S)
\end{align*}
for some distribution $p(u_1,u_2)$ with 
\begin{equation*}
  I(U_1;U_2)\leq C_\mathrm{out}^1
  +C_\mathrm{out}^2,
\end{equation*}
is achievable for the channel $p(y|s,u_1,u_2)$
when no state information is available at
the encoders. Note that every code for this 
channel can be transformed into a code for the channel
$p(y|s,x_1,x_2)$ with causal state information
available at the encoders; for all times $t\in [n]$
and $i\in\{1,2\}$,
simply apply the mapping $f_i$ to the pair
$(U_{it},S_{it})$, where $U_{it}$ is
the output symbol of encoder $i$
and $S_{it}$ is component $i$ of the state
at time $t$.
Note that the new code has the same rate 
and by (\ref{eq:channelTransform}),
the same average error probability as 
the original code.
Thus $\mathscr{C}_T(\mathbf{C}_\mathrm{in}^*,
\mathbf{C}_\mathrm{out})$ contains the set
of all rate pairs $(R_1,R_2)$ satisfying
\begin{align*}
  R_1 &\leq I(U_1;Y|S,U_2)\\
  R_2 &\leq I(U_2;Y|S,U_1)\\
  R_1+R_2 &\leq I(U_1,U_2;Y|S)
\end{align*}
for some distribution $p(u_1,u_2)$ with 
\begin{equation*}
  I(U_1;U_2)\leq C_\mathrm{out}^1
  +C_\mathrm{out}^2
\end{equation*}
and mappings 
\begin{equation*}
f_i:\mathcal{U}_i\times\mathcal{S}_i
\rightarrow\mathcal{X}_i\quad
\text{for }i\in\{1,2\}.
\end{equation*}

To complete the proof, we show that
for every $\delta\geq 0$ and every distribution
\begin{equation*}
  p(u_1,u_2)p(s_1,s_2)
  p(x_1|u_1,s_1)p(x_2|u_2,s_2)
\end{equation*}
satisfying $I(U_1;U_2)\leq \delta$,
there exist alphabets $\mathcal{U}_1'$
and $\mathcal{U}_2'$, mappings
\begin{equation*}
  f_i:\mathcal{U}_i'\times\mathcal{S}_i
  \rightarrow\mathcal{X}_i\quad
  \text{for }i\in\{1,2\},
\end{equation*}
and distribution $p(u_1',u_2')$ such that
$I(U_1';U_2')=I(U_1;U_2)$, and the rate region
calculated with respect to 
\begin{equation*}
  p(u_1',u_2')
  \mathbf{1}\big\{x_1=f_1(u_1',s_1)\big\}
  \mathbf{1}\big\{x_2=f_2(u_2',s_2)\big\},
\end{equation*}
contains the region calculated with respect 
to $p(u_1,u_2)p(x_1|u_1,s_1)p(x_2|u_2,s_2)$.

To this end, applying Lemma \ref{lem:detFunction}, 
which appears at the end of this subsection, to 
$p(x_i|u_i,s_i)$ demonstrates the existence of a
random variable $V_i$ that is independent
of $(U_i,S_i)$ and a mapping 
\begin{equation*}
  f_i:\mathcal{V}_i\times
  \mathcal{U}_i\times\mathcal{S}_i
  \rightarrow\mathcal{X}_i
\end{equation*}
that satisfies
\begin{equation*}
  p(x_i|u_i,s_i)=
  \sum_{v_i}p(v_i)\mathbf{1}\big\{
  x_i=f_i(v_i,u_i,s_i)\big\}.
\end{equation*}
Furthermore, without loss of generality,
we may assume $V_1$ and $V_2$ are independent,
and $(V_1,V_2)$ is independent of 
$(U_1,U_2,S_1,S_2)$.

Let $U'_i\coloneqq (U_i,V_i)$ for $i\in\{1,2\}$.
Then
\begin{align*}
  I(U'_1;U'_2) &= 
  I(U_1,V_1;U_2,V_2)\\
  &= H(U_1,V_1)+H(U_2,V_2)-H(U_1,U_2,V_1,V_2)\\
  &= I(U_1;U_2)+I(V_1;V_2)=I(U_1;U_2).
\end{align*}
We now show that the rate region calculated with 
respect to the distribution 
\begin{equation*}
  p(s_1,s_2)p(u'_1,u'_2)\mathbf{1}\{x_1=f_1(u'_1,s_1)\}
  \mathbf{1}\{x_2=f_2(u'_2,s_2)\}
\end{equation*}
contains the rate region with respect to 
\begin{equation*}
  p(s_1,s_2)p(u_1,u_2)p(x_1|u_1,s_1)p(x_2|u_2,s_2).
\end{equation*}
Recall that $S=(S_1,S_2)$. We have
\begin{align*}
  I(U'_1;Y|S,U'_2)
  &= I(U'_1;Y,U'_2|S)-I(U'_1;U'_2|S)\\
  &= I(U'_1;Y,U'_2|S)-I(U_1;U_2|S)\\
  &= I(U_1,V_1;Y,U_2,V_2|S)-I(U_1;U_2|S)\\
  &\geq I(U_1;Y,U_2|S)-I(U_1;U_2|S)\\
  &= I(U_1;Y|S,U_2).
\end{align*}
Similarly, we show 
\begin{equation*}
  I(U'_2;Y|S,U'_1)\geq 
  I(U_2;Y|S,U_1).
\end{equation*}
Finally, we have
\begin{align*}
  I(U'_1,U'_2;Y|S)
  &= I(U_1,V_1,U_2,V_2;Y|S)\\
  &\geq I(U_1,U_2;Y|S).
\end{align*}
This completes the proof. We next state and prove
Lemma \ref{lem:detFunction}, which we applied earlier
in the proof. 
In Lemma \ref{lem:detFunction}, the scenario where
$\mathcal{X}$ and $\mathcal{S}$ are finite is a
special case of the functional representation
lemma \cite[p. 626]{ElGamalKim}. 

\begin{lem} \label{lem:detFunction}
Let $\{F(\cdot|s)\}_{s\in\mathcal{S}}$
be a collection of cumulative distribution
functions (CDFs) on alphabet 
$\mathcal{X}\subseteq\mathbb{R}$ 
and let $S$ be a random variable with alphabet 
$\mathcal{S}$. 
Then there exists a random variable $U$ independent 
of $S$ and a mapping 
\begin{equation*}
  g:\mathcal{S}\times\mathcal{U}\rightarrow
  \mathcal{X}
\end{equation*}
such that the conditional CDF of $g(S,U)$ given
$S=s$ equals $F(\cdot|s)$. In the case where
$\mathcal{X}$ and $\mathcal{S}$ are finite, we 
can choose $\mathcal{U}$ such that
\begin{equation} \label{eq:Ucardinality}
  |\mathcal{U}|\leq |\mathcal{S}|
  \big(|\mathcal{X}|-1\big)+1. 
\end{equation}
\end{lem}
\begin{proof} We prove the result for general
alphabets $\mathcal{X}\subseteq\mathbb{R}$.  
Let $\mathcal{U}\coloneqq [0,1]$.
Define the mapping 
$g:\mathcal{S}\times\mathcal{U}\rightarrow
\mathcal{X}$ as
\begin{equation*}
  g(s,u)=\inf\Big\{x\in\mathcal{X}\Big|
  F(x|s)\geq u\Big\}.
\end{equation*}
Let $U$ be independent of $S$ and uniformly distributed
on $\mathcal{U}=[0,1]$. From the quantile function theorem
\cite[Theorem 2]{Angus}, it follows 
that for all $s\in\mathcal{S}$, $g(s,U)$ has CDF
$F(\cdot|s)$. Set $X=g(S,U)$. Then 
\begin{align*}
  F_{X|S}(x|s) &= \pr\big\{X\leq x\big|S=s\big\}\\
  &= \pr\big\{g(S,U)\leq x\big|S=s\big\}\\
  &= \pr\big\{g(s,U)\leq x\big\}=F(x|s).\\
\end{align*}
\end{proof}

\subsection{Proof of Theorem \ref{thm:mainMACwState}} 
\label{subsec:mainNoncausal}

From the description of the set 
$\mathcal{C}_\tau(\mathcal{S},\mathcal{X}_1,\mathcal{X}_2,\mathcal{Y})$
in Section \ref{sec:resultMACwState}, we see that it 
suffices to prove Theorem \ref{thm:mainMACwState} only in the 
cases $\tau=0$, $\tau=T$, and $\tau=(\infty,s)$.

\textbf{The case $\tau=0$.} When no state information
is available at the encoders, Theorem \ref{thm:mainMACwState} follows 
by applying \cite[Theorem 3]{kUserMAC} to the MAC
\begin{equation*}
  p(s,y|x_1,x_2)=p(s)p(y|s,x_1,x_2),
\end{equation*}
with input alphabets $\mathcal{X}_1$ and $\mathcal{X}_2$, and 
output alphabet $\mathcal{S}\times\mathcal{Y}$. 

\textbf{The case $\tau=T$.} In this case, it suffices to check
that $p_0(u_1)p_0(u_2)$ satisfies
\begin{equation} \label{eq:p0U1U2Max}
  I_0(U_1,U_2;Y|S)=\max_{p(u_1)p(u_2)}
  I(U_1,U_2;Y|S)
\end{equation}
for the MAC
\begin{equation} \label{eq:u1u2MAC}
  p(y|s,u_1,u_2)=\sum_{x_1,x_2}\mathbf{1}\{x_1=f_1(u_1,s_1)\}
  \mathbf{1}\{x_2=f_2(u_2,s_2)\}p(y|s,x_1,x_2),
\end{equation}
since if (\ref{eq:p0U1U2Max}) holds, then Theorem
\ref{thm:mainMACwState} follows by applying the case
$\tau=0$ to the MAC defined by (\ref{eq:u1u2MAC}).

To prove (\ref{eq:p0U1U2Max}), first note that
for all $p(u_1)p(u_2)$,
\begin{align}
  I(U_1,U_2;Y|S)
  &= H(Y|S)-H(Y|U_1,U_2,S)\notag\\
  &\leq H(Y|S)-H(Y|U_1,U_2,S,X_1,X_2) \label{eq:ineqState1}\\
  &= H(Y|S)-H(Y|S,X_1,X_2) \notag\\
  &= I(X_1,X_2;Y|S) \leq I_0(X_1,X_2;Y|S). \label{eq:ineqState2}
\end{align}
For the distribution $p_0(u_1)p_0(u_2)$, however, the inequalities
in (\ref{eq:ineqState1}) and (\ref{eq:ineqState2}) hold with 
equality due to (\ref{eq:p0XSinTermsU}).

\textbf{The case $\tau=(\infty,s)$.}
In this case, we provide a self-contained proof
as it is not straightforward to derive it from prior
cases. This is due to the fact that in this case,
as described in Lemma \ref{lem:IBnoncausal},
the family of achievable distributions is constrained
by three inequalities rather than one. 

Let $p_0(x_1|s_1)p_0(x_2|s_2)$ be a distribution
that satisfies
\begin{equation*}
  I_0(X_1,X_2;Y|S)
  =\max_{p(x_1|s_1)p(x_2|s_2)}I(X_1,X_2;Y|S).
\end{equation*}
By assumption, there exists a distribution 
$p_1(x_1,x_2|s_1,s_2)$ such that
\begin{equation} \label{eq:assumption1}
  I_1(X_1,X_2;Y|S)+
  \mathbb{E}\Big[D\big(p_1(y|S)\|p_0(y|S)\big)\Big]
  > I_0(X_1,X_2;Y|S),\text{ and}
\end{equation}
\begin{equation} \label{eq:assumption2}
  \forall\: (s_1,s_2)\in\mathcal{S}\colon
  \mathrm{supp}(p_1(x_1,x_2|s_1,s_2))
  \subseteq\mathrm{supp}(p_0(x_1|s_1)p_0(x_2|s_2))
\end{equation}
For every $\lambda\in (0,1)$, define 
\begin{equation*}
  p_\lambda(x_1,x_2|s_1,s_2)
  \coloneqq (1-\lambda)p_0(x_1|s_1)p_0(x_2|s_2)
  +\lambda p_1(x_1,x_2|s_1,s_2).
\end{equation*}
Fix $\epsilon>0$ and $\mathbf{v}\in\mathbb{R}^2_{>0}$. 
Define the mapping $h:[0,1]\rightarrow\mathbb{R}$ as
\begin{equation*} 
  h(\lambda)=\frac{1}{v_1}I_\lambda(X_1;S_2|S_1)
  +\frac{1}{v_2}I_\lambda(X_2;S_1|S_2)
  +\frac{1}{v_1+v_2}
  I_\lambda(X_1;X_2|S_1,S_2)+\epsilon\lambda.
\end{equation*}
A direct calculation, followed by an application of 
(\ref{eq:assumption2}), shows that
\begin{align*}
   \frac{d}{d\lambda}I_\lambda(X_1;S_2|S_1)
   \Big|_{\lambda=0^+}&=0\\
   \frac{d}{d\lambda}I_\lambda(X_2;S_1|S_2)
   \Big|_{\lambda=0^+}&=0\\
   \frac{d}{d\lambda}I_\lambda(X_1;X_2|S_1,S_2)
   \Big|_{\lambda=0^+}&=0.
\end{align*}
Note that $h$ is continuously differentiable and 
\begin{equation*}
  \frac{dh}{d\lambda}\Big|_{\lambda=0^+}
  =\epsilon>0.
\end{equation*}
Therefore, by the inverse function theorem, there exists $h_0>0$
such that $h$ is invertible on $[0,h_0)$; that is, there
exists a mapping $\lambda^*:[0,h_0)\rightarrow [0,1]$ 
that satisfies
\begin{equation} \label{eq:lambdaStar}
  h=\frac{1}{v_1}I_{\lambda^*(h)}(X_1;S_2|S_1)
  +\frac{1}{v_2}I_{\lambda^*(h)}(X_2;S_1|S_2)
  +\frac{1}{v_1+v_2}
  I_{\lambda^*(h)}(X_1;X_2|S_1,S_2)
  +\epsilon\lambda^*(h),
\end{equation}
and
\begin{equation*}
  \frac{d{\lambda^*}}{dh}\Big|_{h=0^+}
  =\frac{1}{\epsilon}.
\end{equation*}
We henceforth write $\lambda^*$ instead of $\lambda^*(h)$
when the value of $h$ is clear from context.
By (\ref{eq:lambdaStar}), it now follows that
for all $h\in [0,h_0)$, 
\begin{align*}
  hv_1 &\geq I_{\lambda^*}(X_1;S_2|S_1)\\
  &=H_{\lambda^*}(X_1|S_1)-H_{\lambda^*}(X_1|S_1,S_2)\\
  hv_2 &\geq I_{\lambda^*}(X_2;S_1|S_2)\\ 
  &=H_{\lambda^*}(X_2|S_2)-H_{\lambda^*}(X_2|S_1,S_2)\\
  h(v_1+v_2)
  &\geq I_{\lambda^*}(X_1;S_2|S_1)
  +I_{\lambda^*}(X_2;S_1|S_2)
  +I_{\lambda^*}(X_1;X_2|S_1,S_2)\\
  &=H_{\lambda^*}(X_1|S_1)+
  H_{\lambda^*}(X_2|S_2)-H_{\lambda^*}(X_2|S_1,S_2).
\end{align*}
Thus, by Lemma \ref{lem:IBnoncausal}, 
\begin{equation} \label{eq:CinftySlower}
  C_{(\infty,s)}(\mathbf{\bar C}_\mathrm{in},h\mathbf{v})
  \geq I_{\lambda^*}(X_1,X_2;Y|S)-
  I_{\lambda^*}(X_1;X_2|S).
\end{equation}
Since equality holds in (\ref{eq:CinftySlower})
at $h=0$, we have
\begin{align}
  \MoveEqLeft
  \liminf_{h\rightarrow 0^+}
  \frac{C_{(\infty,s)}(\mathbf{\bar C}_\mathrm{in},h\mathbf{v})
  -C_{(\infty,s)}(\mathbf{\bar C}_\mathrm{in},\mathbf{0})}{h}\\
  &\geq \frac{1}{\epsilon}\frac{d}{d\lambda^*}
  \Big(I_{\lambda^*}(X_1,X_2;Y|S)-
  I_{\lambda^*}(X_1;X_2|S)\Big)\Big|_{\lambda^*=0^+}
  \notag\\
  &= \frac{1}{\epsilon}\frac{d}{d\lambda^*}
  I_{\lambda^*}(X_1,X_2;Y|S)\Big|_{\lambda^*=0^+}
  \notag\\
  &\geq \frac{1}{\epsilon}\Big(
  I_1(X_1,X_2;Y|S)+
  \mathbb{E}\Big[D\big(p_1(y|S)\|p_0(y|S)\big)\Big]
  -I_0(X_1,X_2;Y|S)\Big). \label{eq:dMIio}
\end{align}
The proof of (\ref{eq:dMIio}) is analogous to
\cite[Lemma 14 (ii)]{kUserMAC} and is omitted. 
Since (\ref{eq:dMIio}) holds for all $\epsilon>0$, 
from (\ref{eq:assumption1}) it follows that 
\begin{equation*}
  \lim_{h\rightarrow 0^+}
  \frac{C_{(\infty,s)}(\mathbf{\bar C}_\mathrm{in},h\mathbf{v})
  -C_{(\infty,s)}(\mathbf{\bar C}_\mathrm{in},\mathbf{0})}{h}
  =\infty.
\end{equation*}

\subsection{Proof of Proposition \ref{prop:gaussianFading}}
\label{subsec:gaussianFading}

Since 
\begin{equation*}
C_{T-1}(\mathbf{C}_\mathrm{in}^*,\mathbf{0})=
C_0(\mathbf{C}_\mathrm{in}^*,\mathbf{0})=
C_0(\mathbf{0},\mathbf{0})
\end{equation*}
and
\begin{equation*}
C_{(\infty,s)}(\mathbf{C}_\mathrm{in}^*,\mathbf{0})=
C_\infty(\mathbf{C}_\mathrm{in}^*,\mathbf{0})=
C_T(\mathbf{C}_\mathrm{in}^*,\mathbf{0})=
C_T(\mathbf{0},\mathbf{0}),
\end{equation*}
it suffices to prove the result only when
$\tau=0$ or $\tau=T$.

When $\tau=0$, from Lemma \ref{lem:IBnoState}, it follows
that for any distribution $p(x_1)p(x_2)$ satisfying 
$\mathbb{E}[X_i^2]\leq P_i$ for $i\in\{1,2\}$
and 
\begin{equation*}
  I(X_1;X_2)\leq C_\mathrm{out}^1+C_\mathrm{out}^2,
\end{equation*}
we have
\begin{align} 
  C_0(\mathbf{C}_\mathrm{in}^*,\mathbf{C}_\mathrm{out})
  &\geq I(X_1,X_2;Y|S)-I(X_1;X_2)\notag\\
  &= H(Y|S)-H(Z)\notag\\
  &= \frac{1}{4}\Big(H(X_1+Z)
  +H(X_2+Z)+H(X_1+X_2+Z)-3H(Z)\Big)
  \label{eq:noStateSumCapacity}
\end{align}

Fix $h>0$. Let $(X_1,X_2)$ be jointly Gaussian 
with mean zero and covariance matrix
\begin{equation*}
  \Sigma \coloneqq 
  \begin{pmatrix}
  \sqrt{P_1} & \rho\sqrt{P_1P_2} \\
  \rho\sqrt{P_1P_2} & \sqrt{P_2}
  \end{pmatrix},
\end{equation*}
where $\rho\in [0,1]$ is chosen such that
\begin{equation*}
  I(X_1;X_2)=\frac{1}{2}\log\frac{1}{1-\rho^2}
  \coloneqq h(v_1+v_2).
\end{equation*}
Then 
\begin{equation*}
  \frac{d\rho}{dh}\Big|_{h=0^+}
  =\infty.
\end{equation*}
Using (\ref{eq:noStateSumCapacity}), it follows
that
\begin{equation*}
  C_0(\mathbf{C}_\mathrm{in}^*,h\mathbf{v})
  -C_0(\mathbf{C}_\mathrm{in}^*,\mathbf{0})
  \geq \frac{1}{8}\log\Big(1+
  \frac{2\rho\sqrt{P_1P_2}}{P_1+P_2+N}\Big)
  -h(v_1+v_2),
\end{equation*}
which implies the desired result. 

A similar proof follows when $\tau=T$. In this
case, for fixed $h>0$, let $(U_1,U_2)$
be jointly Gaussian 
with mean zero and covariance matrix
\begin{equation*}
  \Sigma \coloneqq 
  \begin{pmatrix}
  \sqrt{2P_1} & 2\rho\sqrt{P_1P_2} \\
  2\rho\sqrt{P_1P_2} & \sqrt{2P_2}
  \end{pmatrix},
\end{equation*}
where $\rho\in [0,1]$ satisfies
\begin{equation*}
  I(U_1;U_2)=\frac{1}{2}\log\frac{1}{1-\rho^2}
  \coloneqq h(v_1+v_2).
\end{equation*}
For $i\in\{1,2\}$, set $X_i\coloneqq S_i U_i$.
From Lemma \ref{lem:IBcausal}, it follows that
\begin{align} 
  C_0(\mathbf{C}_\mathrm{in}^*,\mathbf{C}_\mathrm{out})
  &\geq I(U_1,U_2;Y|S)-I(U_1;U_2)\notag\\
  &= H(Y|S)-H(Y|S,U_1,U_2)-I(U_1;U_2)\notag\\
  &= H(Y|S)-H(Y|S,U_1,U_2,X_1,X_2)-I(U_1;U_2)
  \label{eq:XdetU}\\
  &= H(Y|S)-H(Z)-I(U_1;U_2)\notag\\
  &= \frac{1}{4}\Big(H(X_1+Z|S_1=1)
  +H(X_2+Z|S_2=1)\\
  &\phantom{= \frac{1}{4}\Big(}
  +H(X_1+X_2+Z|S_1=1,S_2=1)-3H(Z)\Big)-I(U_1;U_2),
  \label{eq:causalStateSC}
\end{align}
where (\ref{eq:XdetU}) follows from the fact that
$(X_1,X_2)$ is a deterministic function of $(S,U_1,U_2)$.
Simplifying (\ref{eq:causalStateSC}) results in
\begin{equation*}
  C_T(\mathbf{C}_\mathrm{in}^*,h\mathbf{v})
  -C_T(\mathbf{C}_\mathrm{in}^*,\mathbf{0})
  \geq \frac{1}{8}\log\Big(1+
  \frac{4\rho\sqrt{P_1P_2}}{2P_1+2P_2+N}\Big)
  -h(v_1+v_2).
\end{equation*}

\subsection{Outer Bounds in the Absence of Cooperation}
\label{subsec:outerBoundsState}

We next prove outer bounds for 
$\mathscr{C}_{T-1}(\mathbf{0},\mathbf{0})$
and $\mathscr{C}_{\infty}(\mathbf{0},\mathbf{0})$.
Together with our inner bounds in Section \ref{sec:codeState},
these outer bounds determine the capacity region
$\mathscr{C}_\tau(\mathbf{0},\mathbf{0})$ for 
all $\tau$, and show
\begin{equation*}
  \mathscr{C}_0(\mathbf{0},\mathbf{0})
  =\mathscr{C}_{T-1}(\mathbf{0},\mathbf{0})
  \text{ and }
  \mathscr{C}_T(\mathbf{0},\mathbf{0})
  =\mathscr{C}_\infty(\mathbf{0},\mathbf{0})
  =\mathscr{C}_{(\infty,s)}(\mathbf{0},\mathbf{0}).
\end{equation*}
The bounds presented here are well known \cite[p. 175]{ElGamalKim}
and are included for completeness.

For convergent sequences $(a_n)_{n=1}^\infty$
and $(b_n)_{n=1}^\infty$, define notation 
``$\simeq$'' and ``$\lesssim$'' as 
\begin{align*}
  a_n\simeq b_n &:\iff
  \lim_{n\rightarrow\infty}
  \frac{1}{n}(a_n-b_n)=0\\
  a_n\lesssim b_n &:\iff
  \lim_{n\rightarrow\infty}
  \frac{1}{n}(a_n-b_n)\leq 0.
\end{align*}

Consider a sequence of $(2^{nR_1},2^{nR_2},n)$ 
codes with $P_e^{(n)}\rightarrow 0$ as $n\rightarrow\infty$
for the MAC with full state information at the 
decoder. Initially, we do not make any assumptions regarding
the presence of state information at the encoders. 
 
We begin with the bound on $R_1$. We have
\begin{align}
  nR_1 &= H(W_1)\notag\\
  &= H(W_1|S^n,W_2)\notag\\
  &\simeq I(W_1;Y^n|S^n, W_2)\notag\\
  &= H(Y^n|S^n,W_2,X_2^n)-H(Y^n|S^n,W_1,W_2,X_1^n,X_2^n)\notag\\
  &= H(Y^n|S^n,X_2^n)-H(Y^n|S^n,X_1^n,X_2^n)\notag\\
  &=\sum_{t=1}^n \Big(H(Y_t|Y^{t-1},S^n,X_2^n)
  -H(Y_t|Y^{t-1},S^n,X_1^n,X_2^n)\Big).\label{eq:obR1}
\end{align}
Similarly,
\begin{equation*}
  nR_2 \simeq
  \sum_{t=1}^n \Big(H(Y_t|Y^{t-1},S^n,X_1^n)
  -H(Y_t|Y^{t-1},S^n,X_1^n,X_2^n)\Big).
\end{equation*}
Next we bound $R_1+R_2$. We have
\begin{align*}
  n(R_1+R_2) &= H(W_1,W_2)\\
  &= H(W_1,W_2|S^n)\\
  &\simeq I(W_1,W_2;Y^n|S^n)\\
  &= H(Y^n|S^n)-H(Y^n|S^n,W_1,W_2,X_1^n,X_2^n)\\
  &= H(Y^n|S^n)-H(Y^n|S^n,X_1^n,X_2^n)\\
  &=\sum_{t=1}^n \Big(H(Y_t|Y^{t-1},S^n)
  -H(Y_t|Y^{t-1},S^n,X_1^n,X_2^n)\Big).
\end{align*}
To proceed further, we need to apply the 
causality constraints of the state information
at the encoders. 

\textbf{The case $\tau=T-1$.} In this case, strictly
causal state information is available at the encoders; that
is, for $i\in\{1,2\}$ and $t\in [n]$, $X_{it}$ is a deterministic
function of $(W_i,S_i^{t-1})$. Continuing from (\ref{eq:obR1}),
we get
\begin{align*}
  nR_1 &\lesssim 
  \sum_{t=1}^n \Big(H(Y_t|S^t,X_{2t})
  -H(Y_t|S^t,X_{1t},X_{2t})\Big)\\
  &= \sum_{t=1}^n I(X_{1t};Y_t|S^t,X_{2t}).
\end{align*}
Similarly, we get
\begin{align*}
  nR_2 &\lesssim
  \sum_{t=1}^n I(X_{2t};Y_t|S^t,X_{1t})\\
  n(R_1+R_2) &\lesssim
  \sum_{t=1}^n I(X_{1t},X_{2t};Y_t|S^t).
\end{align*}
Setting $Q_t\coloneqq S^{t-1}$ for $t\in [n]$ gives
\begin{align*}
  nR_1 &\lesssim \sum_{t=1}^n I(X_{1t};Y_t|Q_t,S_t,X_{2t})\\
  nR_2 &\lesssim \sum_{t=1}^n I(X_{2t};Y_t|Q_t,S_t,X_{1t})\\
  n(R_1+R_2) &\lesssim 
  \sum_{t=1}^n I(X_{1t},X_{2t};Y_t|Q_t,S_t).
\end{align*}
Thus $\mathscr{C}_{T-1}(\mathbf{0},\mathbf{0})$
is contained in the closure of the set of all rate pairs 
satisfying
\begin{align*}
  R_1 &\leq I(X_1;Y|Q,S,X_2)\\
  R_2 &\leq I(X_2;Y|Q,S,X_1)\\
  R_1+R_2 &\leq I(X_1,X_2;Y|Q,S)
\end{align*}
for some distribution $p(q)p(x_1|q)p(x_2|q)$.

\textbf{The case $\tau=\infty$.}
In this case, noncausal state information is available at
the encoders, meaning that for $i\in\{1,2\}$ and $t\in [n]$,
$X_{it}$ is a deterministic function of $(W_i,S_i^n)$. 
From (\ref{eq:obR1}), we have
\begin{align*}
  nR_1 &\lesssim 
  \sum_{t=1}^n \Big(H(Y_t|S_1^t,S_2^{t:n},X_{2t})
  -H(Y_t|S_1^t,S_2^{t:n},X_{1t},X_{2t})\Big)\\
  &= \sum_{t=1}^n I(X_{1t};Y_t|S_1^t,S_2^{t:n},X_{2t}),
\end{align*}
where for $t\in [n]$,
\begin{equation*}
  S_2^{t:n}=\big(S_{2t},S_{2(t+1)},
  \dots,S_{2n}\big).
\end{equation*}
Similarly, we have
\begin{align*}
  nR_2 &\lesssim
  \sum_{t=1}^n I(X_{2t};Y_t|S_1^t,S_2^{t:n},X_{1t})\\
  n(R_1+R_2) &\lesssim
  \sum_{t=1}^n I(X_{1t},X_{2t};Y_t|S_1^t,S_2^{t:n}).
\end{align*}
For $t\in [n]$, following \cite{PermuterEtAl}, 
define
\begin{equation*}
  Q_t\coloneqq (S_1^{t-1},S_2^{t+1:n}).
\end{equation*}
By assumption, 
$(S_1^n,S_2^n)\overset{\text{iid}}{\sim} p(s_1,s_2)$.
Thus
\begin{align*}
  p(s_1^n,s_2^n|s_1^{t-1},s_2^{t+1:n},s_{1t},s_{2t})
  &= p(s_1^{t+1:n},s_2^{t-1}|s_1^t,s_2^{t:n})\\
  &= p(s_1^{t+1:n}|s_2^{t+1:n})p(s_2^{t-1}|s_1^{t-1}),
\end{align*}
which implies that $S_1^n$ and $S_2^n$ are independent given 
$(Q_t,S_{1t},S_{2t})$. Since $(W_1,W_2)$ is independent
of $(S_1^n,S_2^n)$, it follows that for $t\in [n]$,
$X_{1t}(W_1,S_1^n)$ and $X_{2t}(W_2,S_2^n)$ are 
independent given $(Q_t,S_{1t},S_{2t})$. 
Thus $\mathscr{C}_\infty(\mathbf{0},\mathbf{0})$ is 
contained in the closure of the set of all rate pairs 
satisfying
\begin{align*}
  R_1 &\leq I(X_1;Y|Q,S,X_2)\\
  R_2 &\leq I(X_2;Y|Q,S,X_1)\\
  R_1+R_2 &\leq I(X_1,X_2;Y|Q,S)
\end{align*}
for some distribution $p(q)p(x_1|q,s_1)p(x_2|q,s_2)$.

\section{Conclusion}
The presence of distributed state information in 
a network provides an opportunity for cooperation.
In this work, we study encoder cooperation in the MAC
under the CF model. When no state information is available at
either the encoders or the decoder, \cite{kUserMAC}
provides conditions under which the sum-capacity gain of cooperation
has an infinite slope in the limit of small cooperation rate. This
work extends these conditions to scenarios where distributed 
state information is available at the encoders and 
full state information is available at the decoder. 

\bibliographystyle{IEEEtran}
\bibliography{ref}{}

\end{document}